\documentclass{ifacconf}

\let\theoremstyle\relax

\usepackage{amsmath}
\usepackage{mathtools}
\usepackage{amsthm}
\usepackage{amssymb}
\usepackage{amsfonts}
\usepackage{bm}
\usepackage{balance}
\usepackage{easyReview}
\usepackage{graphicx}

\makeatletter
\def\endfigure{\end@float}
\def\endtable{\end@float}
\makeatother

\let\ifacconfcaptionwidth\captionwidth
\usepackage[caption=false]{subfig}
\let\captionwidth\ifacconfcaptionwidth

\newtheorem{definition}{Definition}[section]
\newtheorem{theorem}{Theorem}
\newtheorem{lemma}{Lemma}
\newtheorem{corollary}{Corollary}

\usepackage{graphicx}      
\usepackage{natbib}

\begin{document}
\begin{frontmatter}

\title{Decentralized Vehicle Coordination and Lane Switching without Switching of Controllers\thanksref{footnoteinfo}} 

\thanks[footnoteinfo]{This work was supported by the ERC Consolidator Grant
LEAFHOUND, the Swedish Research Council (VR) and the Knut
och Alice Wallenberg Foundation (KAW).}

\author{Arno Frauenfelder},
\author{Adrian Wiltz} and
\author{Dimos V. Dimarogonas}

\address{The authors are with the Division of Decision and Control Systems,
KTH Royal Institute of Technology, SE-100 44 Stockholm, Sweden $\{$arnof, wiltz,
dimos$\}$@kth.se.}

\begin{abstract}                
This paper proposes a controller for safe lane change manoeuvres of autonomous vehicles using 
high-order control barrier and Lyapunov functions. The inputs are
calculated using a quadratic program (CLF-CBF-QP) which admits short calculation times. The controller allows for adaptive cruise control, lane following, lane switching and ensures collision
avoidance at all times. The novelty of the controller is the decentralized approach to the coordination of vehicles without switching of controllers. In particular, vehicles indicate their manoeuvres which influences their own safe region and that of neighboring vehicles. This is achieved by introducing so-called coordination functions in the design of control barrier functions. In a relevant simulation example, the controller is validated and its effectiveness is demonstrated.
\end{abstract}


\end{frontmatter}

\section{INTRODUCTION}
\label{sec:Intro}
The automotive industry evolves towards autonomous vehicles, promising more energy efficient travels, reduced accidents due to the elimination of human error, and higher traffic efficiency. 
Following the taxonomy proposed in \cite{Mariani2020}, lane changing (or ramp merging) can be classified as competitive and task-oriented, and is one of the key coordination problems for autonomous vehicles, for which several approaches are proposed with differing amount of decision autonomy for the individual vehicles.
Approaches with a centralized controller are often applied in a platooning scenario, where vehicles communicate with each other (V2V) or with a coordinator. In \cite{Lu2003} and \cite{Torres2017}, an automated ramp merging manoeuvre is proposed with V2V communication and a centralized coordinator.
In \cite{Awal2013}, a centralized approach is used but all calculations are carried out on a leader vehicle.
In \cite{SCHOLTE2022103511}, the platoon merging is proposed without coordinator. Instead, a combination of pre-existing platooning controllers and MPC is used.
In \cite{Werling2010}, an optimal trajectory is planned for vehicle following, velocity keeping and collision avoidance. Other works explore reinforcement learning for decision making and control of lane changing situations \cite[]{Shi2019}.

In order to ensure the satisfaction of safety constraints by means of ensuring the invariance of the set of "safe states", we use Control Barrier Functions (CBF) \cite[]{Wieland2007}. 
Control Lyapunov Functions (CLF) can be used in combination with CBFs in a quadratic program (CLF-CBF-QP) \cite[]{Romdlony2014}. Whereas CBFs work as a safety filter and guarantee the satisfaction of safety constraints, CLFs ensure asymptotic stability. For constraints concerning states that cannot be "directly" controlled, so-called higher order constraints, \cite{Tan2021} and \cite{Xiao2021} propose high-order control barrier functions (HOCBF).

In the context of vehicle coordination, approaches based on the combination of CBF and CLF have been considered in multiple works. A CLF-CBF-QP is used in \cite{Ames2014a} for adaptive cruise control. In \cite{he2021lane-change-cbf}, the CLF-CBF-QP approach is used in combination with a rule-based control strategy. In \cite{Xiao2021a}, CLFs and CBFs are used in an optimization problem to find a collision-free trajectory that leads to the least violation in a rule priority structure.

In this paper, we propose a novel coordination approach for lane switching and adaptive cruise control. It is based on the assumption made by one vehicle that neighboring vehicles behave in a particular way. In return, the vehicle guarantees that it will exhibit the same behavior towards its neighbors. In order to account for the complexity of the lane switching task, coordination functions are introduced such that CLF-CBF-QP approaches become applicable.
The proposed control strategy is completely decentralized and only relies on sensor measurements or V2V communication. It can be combined with a high level traffic coordinator that prescribes reference velocities and times for lane switching. However, this is not necessary for the provided safety guarantees.

The remainder is as follows. Sec.~\ref{sec:Prelim} reviews CBFs and CLFs and introduces their higher order version.
Sec.~\ref{sec:CtrlAppr} presents the control approach and derives safety guarantees. The controller's effectiveness is demonstrated with a simulation in Sec.~\ref{sec:Sim}. Concluding remarks are given in Sec.~\ref{sec:Con}.

\section{PRELIMINARIES}
\label{sec:Prelim}
We consider an \emph{input-affine system}
\begin{equation}
    \dot{\bm{x}} = f(\bm{x}) + g(\bm{x})\bm{u},
    \label{eqn:affineSys}
\end{equation}
with initial condition $\bm{x}(t_0)=\bm{x}_0$, $\bm{x} \in \mathcal{X} \subseteq \mathbb{R}^n$, $\bm{u} \in \mathcal{U} \subseteq \mathbb{R}^m$ and $\mathcal{X}$, $\mathcal{U}$ denote the state and input space, respectively. The functions $f:\mathcal{X} \rightarrow \mathbb{R}^n$ and $g:\mathcal{X} \rightarrow \mathbb{R}^{n \times m}$ are continuous and locally Lipschitz.
A \emph{class $\mathcal{K}$ function} $\alpha: \mathbb{R}_{\geq0} \rightarrow \mathbb{R}_{\geq0}$ is a continuous and strictly increasing function with $\alpha(0) = 0$ \cite[]{Khalil2015}. Furthermore, an extended class $\mathcal{K}$ function $\gamma: \mathbb{R} \rightarrow \mathbb{R}$ is continuous and strictly increasing function with $\gamma(0) = 0$. Since higher-order systems are considered, the \emph{relative degree $r$} of a function is defined.
\begin{definition} (Relative degree \cite{Khalil2015}): Let $b: \mathcal{X} \rightarrow \mathbb{R}$ be a $r^{th}$-order differentiable function. The function $b$ has relative degree $r$ on $\mathcal{X}$ with respect to \eqref{eqn:affineSys} if
\begin{equation}
\begin{aligned}
    L_g\,L_f^{i}\,b(\bm{x}) & = 0, \quad i = 1, 2, ..., r-2;\\ 
    L_g\, L_f^{r-1}\,b(\bm{x}) & \neq 0
\end{aligned}
\end{equation}
for all $\bm{x} \in \mathcal{X}$.
\label{def:reldeg}
\end{definition}
$L_f b(\bm{x})$ and $L_g b((\bm{x}))$ denote the Lie derivatives of $b(\bm{x})$ along the vector fields $f$ and $g$, respectively. 

\subsection{High-Order Barrier Functions}
In order to ensure safety, the concept of \emph{High-Order Control Barrier Functions (HOCBF)} is introduced. For a differentiable function $b: \mathcal{X} \rightarrow \mathbb{R}$, we define the superlevel set $\mathcal{C}$ as 
\begin{equation}
    \mathcal{C}:= \{\bm{x} \in \mathcal{X}\:|\: b(\bm{x}) \geq 0\}.
    \label{eqn:CSets}
\end{equation}
Moreover, we define $\psi_i : \mathcal{X} \times \mathcal{T} \rightarrow \mathbb{R}$, $i = \{1,2,...,r\}$, and $\mathcal{T} = [t_1,t_2] \subseteq \mathbb{R}$ as closed time interval, as
\begin{subequations}
\label{eqn:psiHOBF}
\begin{align}
    \psi_0(\bm{x}) & = b(\bm{x}),\label{eqn:psiHOBFa}\\
    \psi_i(\bm{x}) & = \dot{\psi}_{i-1}(\bm{x})+\gamma_i(\psi_{i-1}(\bm{x})),\label{eqn:psiHOBFb}
\end{align}
\end{subequations}
where $\gamma_i$ is an extended class $\mathcal{K}$ function. For each function $\psi_i(\bm{x})$, the corresponding set $\mathcal{C}_i$, $i = \{1,2,...,r\}$, is defined as
\begin{equation}
    \mathcal{C}_i := \{\bm{x} \in \mathcal{X} \:|\: \psi_{i-1} \geq 0\}.
    \label{eqn:CiSets}
\end{equation}
\begin{definition}[High-Order Control Barrier Function (HOCBF)]
Let functions $\psi_i(\bm{x})$ and sets $\mathcal{C}_i$, $i=\{1,2,...,r-1\}$, be defined by~\eqref{eqn:psiHOBFb} and~\eqref{eqn:CiSets}, respectively.
The differentiable function $b: \mathcal{X} \rightarrow \mathbb{R}$ with relative degree $r$ is a \emph{High-Order Control Barrier Function (HOCBF)} for control system \eqref{eqn:affineSys}, if there exists an extended class $\mathcal{K}$ function $\gamma_r$ such that for all $\bm{x} \in \mathcal{X}$
\begin{multline}
    \sup_{\bm{u}\in \mathcal{U}}[L_f\,\psi_{r-1}(\bm{x})+L_g \,\psi_{r-1}(\bm{x})\, \bm{u}] \geq -\gamma_r(\psi_{r-1}(\bm{x})).\label{eqn:HOCBF}
\end{multline}
\label{def:HOCBF}
\end{definition}
If $r=1$, we call $b(\bm{x})$ a control barrier function (CBF) \cite[]{Ames2017}.
The above definition leads us to input sets
\begin{multline}
    \mathcal{U}_{\text{HOCBF}}(\bm{x}) := \{ u \in \mathcal{U} \:|\: L_f\,\psi_{r-1}(\bm{x})+L_g \,\psi_{r-1}(\bm{x})\, \bm{u}\\\geq -\gamma_r(\psi_{r-1}(\bm{x}))\}.\label{eqn:UAHOCBF}
\end{multline}
\begin{theorem}
\cite[]{Tan2021}: Consider an HOCBF $b$ for control system~\eqref{eqn:affineSys}. Then any locally Lipschitz continuous control $\bm{u}(\bm{x}) \in \mathcal{U}_{\text{HOCBF}}(\bm{x})$ applied to system \eqref{eqn:affineSys} renders the set $\mathcal{C} := \bigcap_{i=1}^r \mathcal{C}_i$ forward invariant and asymptotically stable.
\label{thm:CifwdInv}
\end{theorem}

\subsection{High Order Control Lyapunov Function}
Analogously to HOCBFs, we introduce \emph{High Order Control Lyapunov Functions (HOCLF)}, in order to ensure asymptotic stability.
We define a series of functions $\eta_i : \mathcal{X} \rightarrow \mathbb{R}$, $i = \{1,2,...,r\}$,
\begin{subequations}
 \label{eqn:etaHOLF}
    \begin{align}
    \eta_0(\bm{x}) & = -\dot{V}(\bm{x})-\alpha(V(\bm{x})),\label{eqn:etaHOLFa}\\
    \eta_i(\bm{x}) & = \dot{\eta}_{i-1}(\bm{x})+\mu_i(\eta_{i-1}(\bm{x})),\label{eqn:etaHOLFb}
    \end{align}
\end{subequations}
where $V(\bm{x})$ is a Lyapunov function \cite[]{Khalil2015}, $\alpha$ is a class $\mathcal{K}$ function and $\mu_i$ are extended class $\mathcal{K}$ functions. Superlevel sets $\mathcal{S}_i$, $i=\{1,2,...,r\}$, are defined as
\begin{equation}
    \mathcal{S}_i := \{ \bm{x} \in \mathcal{X} \:|\: \eta_{i-1}(\bm{x}) \geq 0 \}.
    \label{eqn:SiSets}
\end{equation}
\begin{definition}[High-Order Control Lyapunov Function (HOCLF)] \label{def:HOCLF}
Let functions $\eta_i(\bm{x})$ and sets $\mathcal{S}_i$, $i = \{1,2,...,r\}$, be defined by~\eqref{eqn:etaHOLF} and~\eqref{eqn:SiSets}, respectively. A differentiable function $V: \mathcal{X} \rightarrow \mathbb{R}_{\geq 0}$ with
$V(0) = 0$ and $V(\bm{x})>0$ for all  $x \in \mathcal{X}$, $\bm{x}\not = 0$,  
is a \emph{HOCLF} for \eqref{eqn:affineSys} if there exist differentiable extended class $\mathcal{K}$ functions $\mu_i$, $i = \{1,2,...,r\}$, and a class $\mathcal{K}$ function $\alpha$ such that
\begin{subequations}
\label{eqn:HOCLF2}
\begin{align}
    \dot{V}(\bm{x}) & \leq -\alpha(V(\bm{x})),\label{eqn:HOCLF2a}\\
    \sup_{\bm{u}\in \mathcal{U}}[L_f\,\eta_{r-1}(\bm{x})+L_g \,\eta_{r-1}(\bm{x})\, \bm{u}] & \geq -\mu_r(\eta_{r-1}(\bm{x})).\label{eqn:HOCLF2b}
    \end{align}
\end{subequations}
\end{definition}
If $r=1$, we call $V(\bm{x})$ Control Lyapunov Function (CLF) \cite[]{Khalil2015}.
For an HOCLF, we define the input set
\begin{multline}
\label{eqn:UAHOCLF}
    \mathcal{U}_{\text{HOCLF}}(\bm{x}) = \{ u \in \mathcal{U} \:|\: L_f\,\eta_{r-1}(\bm{x})+L_g \,\eta_{r-1}(\bm{x})\, \bm{u}\\\geq -\mu_r(\eta_{r-1}(\bm{x}))\}.
\end{multline}
\begin{theorem}
Let the origin be an equilibrium point of~\eqref{eqn:affineSys}, i.e., $0=f(0)+g(0)\bm{u}$ for some $\bm{u}\in\mathcal{U}$. Furthermore, let $V(\bm{x})$ be a HOCLF and $\dot{V}(\bm{x}(0))\leq -\alpha(V(\bm{x}(0)))$. Then, any locally Lipschitz continuous controller $\bm{u}(\bm{x}) \in \mathcal{U}_{\text{HOCLF}}(\bm{x})$ asymptotically stabilizes the origin.
\label{thm:HOCLFasympStab}
\end{theorem}
\begin{proof}
As $\dot{V}(\bm{x}(0))\leq -\alpha(V(\bm{x}(0)))$ holds initially, the set where \eqref{eqn:HOCLF2a} holds is rendered invariant by \eqref{eqn:HOCLF2b} according to Theorem \ref{thm:CifwdInv}. Then, asymptotic stability follows from \cite[Thm. 3.3]{Khalil2015}.
\end{proof}

\subsection{Optimization Problem}
Altogether, the \emph{optimization problem} for computing input~$\bm{u}$ is given as a quadratic program (QP)
\begin{subequations}
\label{eqn:QP1}
\begin{align}
\min_{\bm{u} \:\in\: \mathcal{U},\;\delta \:\in\: \mathbb{R}} \quad & 
\frac{1}{2}\bm{u}^T H \bm{u} + p \cdot \delta^2\\
\textrm{s.t.} \quad & 
\begin{aligned}
   L_f\,\psi_{r-1}(\bm{x})+L_g \,\psi_{r-1}(\bm{x})\, \bm{u}\\\geq -\gamma_r(\psi_{r-1}(\bm{x})) \label{eqn:QP1b}\end{aligned}\\
  & \begin{aligned}
  L_f\,\eta_{r-1}(\bm{x})+L_g \,\eta_{r-1}(\bm{x})\, \bm{u}\\\geq -\mu_r(\eta_{r-1}(\bm{x})) + \delta, \label{eqn:QP1c}
  \end{aligned}
\end{align}
\end{subequations}
with a positive-definite matrix $H \in \mathbb{R}^{m \times m}$, a scalar $p > 0$. Constraint \eqref{eqn:QP1b} is a safety constraint; \eqref{eqn:QP1c} is a stabilization constraint relaxed with a slack variable $\delta$. 

Next, we generalize the QP for multiple HOCBFs $b_k \in \{b_1, b_2, ..., b_q\}$, $q\geq1$, and HOCLFs $V_j \in \{V_1,V_2, ... , V_h\}$, $h\geq1$, such that we can take multiple objectives into account. For each $b_k$, we denote the input set~\eqref{eqn:UAHOCBF} as $\mathcal{U}_{k,\text{HOCBF}}$, and for each $V_j$ the input set~\eqref{eqn:UAHOCLF} as $\mathcal{U}_{j, \text{HOCLF}}$. Similarly for each function $b_k$, the corresponding functions $\psi_i$, $i=\{1,...,r\}$, as defined in~\eqref{eqn:psiHOBF} are denoted by $\psi_{k,i}$, and for each function~$V_j$, the corresponding functions $\eta_i$, $i=\{1,...,r\}$, as defined in~\eqref{eqn:etaHOLF} are denoted by~$\eta_{j,i}$.
Moving all terms in \eqref{eqn:QP1b}-\eqref{eqn:QP1c} to the left-hand side and summarizing them in stack vectors yields
\begin{subequations}
\label{eqn:mult}
\begin{align}
    \psi(\bm{x}) &:= \begin{bmatrix}
    L_f\,\psi_{1,r-1}(\bm{x})+L_g \,\psi_{1,r-1}(\bm{x})\, \bm{u}\\+\gamma_r(\psi_{1,r-1}(\bm{x}))\\
    \vdots\\
    L_f\,\psi_{q,r-1}(\bm{x})+L_g \,\psi_{q,r-1}(\bm{x})\, \bm{u}\\+\gamma_r(\psi_{q,r-1}(\bm{x}))\\
    \end{bmatrix},
    \label{eqn:psiconj}
\\
    \eta(\bm{x},\bm{\delta}) &:= \begin{bmatrix}
    L_f\,\eta_{1, r-1}(\bm{x})+L_g \,\eta_{1,r-1}(\bm{x})\, \bm{u} \\+\mu_{1,r}(\eta_{1,r-1}(\bm{x}))-\delta_1\\
    \vdots\\
    L_f\,\eta_{h, r-1}(\bm{x})+L_g \,\eta_{h,r-1}(\bm{x})\, \bm{u} \\+\mu_{h,r}(\eta_{h,r-1}(\bm{x}))-\delta_h\\
    \end{bmatrix}, \label{eqn:etaconj}
\end{align}
\end{subequations}
where $\bm{\delta} = [\delta_1, ... , \delta_h]^T$ is the vector of all slack variables.
Then, the QP with multiple HOCLFs and HOCBFs is
\begin{subequations}
\label{eqn:QP2}
\begin{align}
\min_{\bm{u} \:\in\: \mathcal{U},\;\delta \:\in\: \mathbb{R}} \quad & 
\frac{1}{2}\bm{u}^T H \bm{u} + \frac{1}{2}\bm{\delta}^T P \bm{\delta}\\
\textrm{s.t.} \quad & 
    \psi(\bm{x})\geq 0\\
    & \eta(\bm{x},\bm{\delta}) \geq 0,
\end{align}
\end{subequations}
with positive-definite matrices $P \in \mathbb{R}^{h \times h}$ and $H \in \mathbb{R}^{m \times m}$.
In order to exclude contradicting objectives, we assume that $\mathcal{C} := (\bigcap_{k=1}^q \mathcal{C}_k) \neq \emptyset$ and $\mathcal{U}_{\text{HOCBF}}(\bm{x}) := (\bigcap_{k=1}^q \mathcal{U}_{k,\text{HOCBF}}(\bm{x})) \neq \emptyset$, $\forall \bm{x} \in \mathcal{C}$. In \cite{Tan2022}, such HOCBFs are called 
\emph{compatible}.
Analogously, we say that HOCLFs $V_j$ are compatible if $\mathcal{S} := (\bigcap_{j=1}^h \mathcal{S}_j) \neq \emptyset$ and $\mathcal{U}_{\text{HOCLF}}(\bm{x}) := (\bigcap_{j=1}^h \mathcal{U}_{j,\text{HOCLF}}(\bm{x})) \neq \emptyset$ for all $\bm{x} \in \mathcal{S}$.
As a direct consequence of Thm. \ref{thm:CifwdInv} and \ref{thm:HOCLFasympStab}, we obtain the following result.
\begin{corollary}
Consider the optimization problem \eqref{eqn:QP2}. If the (HO)CBFs $b_k$, $k=\{1,...,q\}$, are compatible, then $\bm{u}(\bm{x})$ minimizing \eqref{eqn:QP2} renders $\mathcal{C}$ forward invariant and asymptotically stable on $\mathcal{X}$. 
Let the origin be an equilibrium point of \eqref{eqn:affineSys}.
If additionally $\bm{\delta}=0$ and (HO)CLF in \eqref{eqn:etaconj} are compatible, i.e., $\mathcal{S} \neq \emptyset$ and $\mathcal{U}_{\text{HOCLF}}(\bm{x}) \neq \emptyset$, then the origin is asymptotically stable.
\label{cor:QP}
\end{corollary}

\subsection{Control Problem}
\label{sec:CtrlProb}
We consider a road with three lanes, each of width $w$, and vehicles as depicted in Fig.~\ref{fig:CtrlProb}. The neighboring vehicles of the ego-vehicle (E) are denoted with indices according to their position relative to the ego-vehicle. The ego-vehicle determines the position of neighbouring vehicles through its sensors (with omnidirectional sensor range $r_S$) or V2V communication and decides whether they are in front (F), in the back (B), on the same lane (0), on a lane to the right (-1) or on a lane to the left (+1)\footnote{For neighbouring vehicles with identical $x$-coordinate as the ego-vehicle: vehicles on the left lane (+1) are considered to be in front~(F), vehicles on the right lane (-1) are considered to be in the back~(B).}. For example, index +1B denotes a vehicle behind the ego-vehicle on the lane to the left. If in one of these positions no vehicle is within the sensor range, a mock vehicle is placed at distance~$r_S$. Thereby the worst case is assumed that a neighboring vehicle might be located just outside of the sensor range. Moreover, each vehicle has the objective to follow a lane. Each lane is denoted by an integer $\ell \in \lbrace 1,2,3\rbrace$. All vehicles are modeled as (nonholonomic) unicycles with dynamics
\begin{equation}
    \label{eqn:VehDyn}
    \begin{cases}
    \dot{x} = v \cos{(\Psi)}\\
    \dot{y} = v \sin{(\Psi)}\\
    \dot{\Psi} = \omega
    \end{cases}
\end{equation}
with states $[x,y,\Psi]^T$, where $x$, $y$ denote the vehicle's position, $\Psi$ its orientation, and inputs $\bm{u} = [v,\omega]^T$, where $v$ denotes the vehicle's longitudinal velocity and $\omega$ its angular velocity. We denote the states of the ego-vehicle as $\bm{x}_E := [x_E,y_E,\Psi_E]$, of vehicle 0F as $\bm{x}_{0F} := [x_{0F},y_{0F},\Psi_{0F}]$ and correspondingly for the other neighboring vehicles. Furthermore, we define the set of all indices of neighboring vehicles as $\mathcal{N} := \{{+1F}, {+1B}, {0F}, {0B}, {-1F}, {-1B}\}$; the states of the ego-vehicle and its neighbors as $\bm{x}:= [\bm{x_E}^T,\bm{x}^T_{i \in \mathcal{N}}]^T = [\bm{x}^T_E$, $\bm{x}^T_{+1F}$, $\bm{x}^T_{+1B}$, $\bm{x}^T_{0F}$, $\bm{x}^T_{0B}$, $\bm{x}^T_{-1F}$, $\bm{x}^T_{-1B}]^T$.
\begin{figure}[t]
    \centering
    \includegraphics[width=\columnwidth]{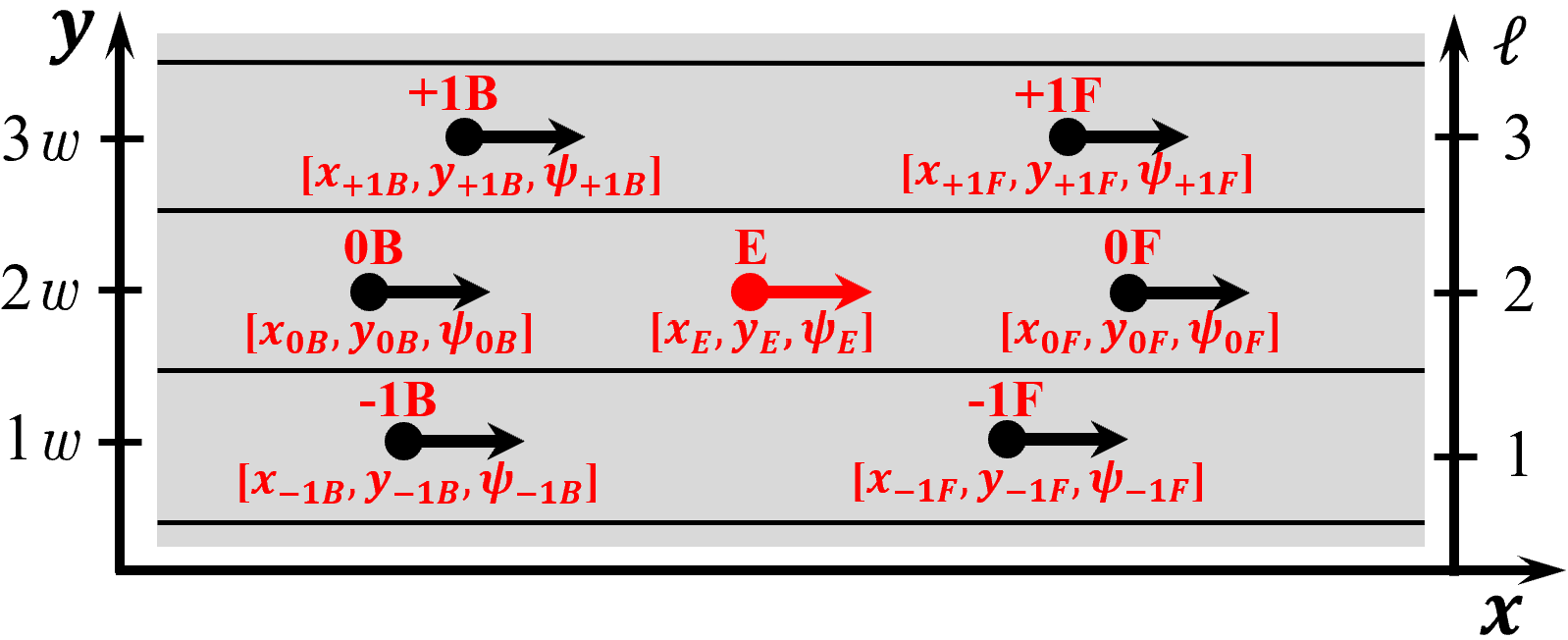}
    \caption{Road with three lanes with width $w$, denoted by $\ell \in \lbrace 1,2,3\rbrace$: indices and states of the ego-vehicle and neighboring vehicles.}
    \label{fig:CtrlProb}
\end{figure}

In this paper, the objective is to develop a decentralized control strategy for
\begin{itemize}
    \item safe lane switching and lane following: ensure a safe distance along the $x$-coordinate between ego-vehicle~E and neighbouring vehicles $i \in \mathcal{N}$ ($\min_{i\in \mathcal{N}}$ $|x_E-x_i|>d$, with safety distance $d\geq0$) and manoeuvre ego-vehicle E to a desired lane ($\lim_{t \to \infty} y_E = y_{\text{ref}}$);
    \item adaptive cruise control (ACC): follow a vehicle in a safe distance ($\min_{i\in \mathcal{N}}$ $|x_E-x_i|>0$) and adjust the velocity ($\lim_{t\to \infty} v_E = \min\{v_{\text{ref}}, v_{0F}\}$).
\end{itemize}
The longitudinal reference velocity $v_{\text{ref}}$ and the vertical reference position on the lane $y_{\text{ref}}$ are set by the driver or a high level traffic coordinator.

\section{CONTROL APPROACH}
\label{sec:CtrlAppr}
\subsection{Reference tracking}
We use an HOCLF to steer the vehicle to the reference lane $y_{\text{ref}}$ and follow it.
A candidate for such an HOCLF is
\begin{equation}
    V(\bm{x}) = \frac{1}{2}(y_{\text{ref}}-y_E)^2.
    \label{eqn:V1}
\end{equation}
\begin{lemma}
    The function $V(\bm{x})$ is a HOCLF for \eqref{eqn:VehDyn} with relative degree $r=2$.
\end{lemma}
\begin{proof}
$V(\bm{x})$ is differentiable, $V(y_{\text{ref}}) = 0$ and $V(\bm{x})>0$, $\forall \bm{x} \in \mathcal{X}$, $y_E \neq y_{\text{ref}}$. Based on \eqref{eqn:etaHOLFa}, we get $\eta_{1,0}(\bm{x}) = (y_{\text{ref}}-y_E) v_E \sin{(\Psi_E)} - \alpha(V(\bm{x}))$. Since $L_g V(\bm{x}) = 0$ for $\Psi_E = 0$, it is $r \neq 1$. However, $L_g L_f V(\bm{x}) \neq 0$ $\forall \bm{x} \in \mathcal{X}$, hence $r=2$, and $V(\bm{x})$ is a HOCLF based on Def. \ref{def:HOCLF}.
\end{proof}
As the vehicle's velocity is a control input, we can directly incorporate the objective on the reference velocity $v_{\text{ref}}$ by
\begin{equation}
    v_E-v_{\text{ref}} = 0.
\label{eqn:vConstr}
\end{equation}
Since safety constraints, which are introduced next, have always precedence over other control objectives, we relax stability and tracking constraints \eqref{eqn:V1} and \eqref{eqn:vConstr} below with slack variables.

\begin{table*}[t]
	\centering
	\begin{tabular}{|c|c|c|} 
		\hline
		\bf{Barrier function} & \bf{Type} & \bf{Description} \\ [0.5ex]
		\hline\hline
		$b_1(\bm{x}) = x_{0F}-x_E-\tau_D v_E$ & CBF & Keeping distance in $x$-direction to 0F.\\ 
		\hline
		$b_{2}(\bm{x}) = y_E -y_{\text{min}}^\ell+ w \cdot \lambda(\theta(x_E,x_{-1B},v_{-1B}))$ & HOCBF & y lower bound based on distance to -1B.\\
		\hline
		$b_{3}(\bm{x}) = y_E -y_{\text{min}}^\ell+ w \cdot \lambda(\theta(x_{-1F},x_E,v_E))$
		& HOCBF & y lower bound based on distance to -1F. \\
		\hline
		$b_{4}(\bm{x}) = w\cdot\lambda(\theta(x_E,x_{+1B},v_{+1B}))+y_{\text{max}}^\ell-y_E$
		& HOCBF & y upper bound based on distance to +1B.\\
		\hline
		$b_{5}(\bm{x}) = w \cdot \lambda(\theta(x_{+1F},x_E,v_E))+y_{\text{max}}^\ell-y_E$ & HOCBF & y upper bound based on distance to +1F.\\
		\hline
		$b_6(\bm{x}) = x_{-1F}-x_E-\tau_D v_E \cdot \sigma(\rho(y_E,y_{-1F}))$ & CBF & Keeping distance in $x$-direction to -1F.\\
		\hline
		$b_7(\bm{x}) = x_{+1F}-x_E-\tau_D v_E \cdot \sigma(\rho(y_{+1F},y_E))$ & CBF & Keeping distance in $x$-direction to +1F.\\
		\hline
	\end{tabular}
	\caption{Overview of barrier functions.}
	\label{tab:barrierfct}
	\vspace{-0.2cm}
\end{table*}

\subsection{Construction of safety constraints}
\paragraph*{Safe distance keeping to preceding vehicle:} 
We choose the continuous differentiable candidate CBF
\begin{equation}
    b_1(\bm{x}) = x_{0F}-x_E-\tau_D v_E.
    \label{eqn:CBFb1}
\end{equation}
Here, $\tau_D v_E$ is a safety distance which depends on the velocity of the ego-vehicle $v_E$ and a time constant $\tau_D > 0$.
\paragraph*{Safe distance keeping to vertically neighboring vehicles:} 
To this end, we introduce a strictly increasing, continuously differentiable function $\lambda: \mathbb{R} \rightarrow [0,1]$ as well as a continuously differentiable function $\theta: \mathcal{X} \times \mathcal{X} \times \mathcal{X} \rightarrow \mathbb{R}_{\geq0}$. The function $\theta$ is the input to $\lambda$ and defined as
\begin{equation}
    \theta(x_1,x_2,v_2) := \frac{x_1-x_2}{\tau_D v_2},
    \label{eqn:theta}
\end{equation}
where $x_1$ and $x_2$ are $x$-coordinates of two distinct vehicles with $x_1>x_2$, and $v_2$ is the velocity of the second vehicle.
$\theta$~can be viewed as a percentage of safety distance $\tau_D v_{2}$.
Function $\lambda$ has the following properties:
\begin{subequations}
    \begin{align}
    \theta(x_1,x_2,v_2) = 0 &\Rightarrow \lambda(\theta(x_1,x_2,v_2)) = 0 \label{eqn:lambda1}\\
    \theta(x_1,x_2,v_2) = 0.9 &\Rightarrow \lambda(\theta(x_1,x_2,v_2)) = 0.5 \label{eqn:lambda2}\\
    \theta(x_1,x_2,v_2) \geq 1 &\Rightarrow 1 \leq \lambda(\theta(x_1,x_2,v_2)) \leq 1.01. \label{eqn:lambda3}
\end{align}
\end{subequations}

$\lambda$ can be viewed as a percentage of the lane width. 
A feasible choice of $\lambda$, fulfilling the assumptions above, is 
\begin{equation}
    \begin{aligned}
    \lambda(\theta(x_1,x_2,v_2)) := \hspace{5.7cm} \\
    \begin{cases}
    \frac{0.5}{0.9} \cdot \theta(x_1,x_2,v_2) & \text{if} \, \theta(x_1,x_2,v_2) \leq 0.9\\
    a_1 (\theta(x_1,x_2,v_2)+a_2)^3+a_3 & \text{if} \, 0.9 < \theta(x_1,x_2,v_2) \leq 1\\
    \frac{1}{1+e^{-\beta_1 (\theta(x_1,x_2,v_2)+\beta_2)}}+\beta_3 & \text{if} \, \theta(x_1,x_2,v_2) > 1,
    \end{cases}
    \end{aligned}
    \label{eqn:lambda}
\end{equation}
with design parameters $a_1, a_3, \beta_1, \beta_3 \in \mathbb{R}_{>0}$ and $a_2, \beta_2 \in \mathbb{R}_{<0}$, chosen such that $\lambda$ is differentiable.
Respectively for each vehicle -1B, -1F, +1B, +1F, we get candidate HOCBFs
\begin{subequations}\begin{align} 
        b_{2}(\bm{x}) & := y_E -y_{\text{min}}^\ell+ w \cdot \lambda(\theta(x_E,x_{-1B},v_{-1B})), \label{eqn:b2}\\
        b_{3}(\bm{x}) & := y_E -y_{\text{min}}^\ell+ w \cdot \lambda(\theta(x_{-1F},x_E,v_E)), \label{eqn:b3}\\
        b_{4}(\bm{x}) & := w \cdot \lambda(\theta(x_E,x_{+1B},v_{+1B}))+y_{\text{max}}^\ell-y_E, \label{eqn:b4}\\
        b_{5}(\bm{x}) & := w \cdot \lambda(\theta(x_{+1F},x_E,v_E))+y_{\text{max}}^\ell-y_E,\label{eqn:b5}
    \end{align}\label{eqn:b2b3}\noindent 
\end{subequations}
where $\ell$ denotes the lane and $w$ the lane width; the lower and upper bound of a lane $\ell$ are $y_{\text{min}}^\ell:=w \cdot \ell-\frac{w}{2}+\epsilon$ and $y_{\text{max}}^\ell:=w \cdot \ell+\frac{w}{2}-\epsilon$, with $\epsilon \geq 0$. The parameter $\epsilon$ is introduced in order to prevent collisions exactly at the middle line between two lanes. We show in Lemma~\ref{thm:HOCBF} below that \eqref{eqn:b2b3} indeed are valid HOCBFs.

\paragraph*{Safe distance keeping to vehicles $\pm$1F:} 
Similar to function $\lambda$ before, we introduce a strictly decreasing, continuous and differentiable function $\sigma : \mathbb{R} \rightarrow [0,1]$, as well as a continuous and differentiable function $\rho: \mathcal{X} \times \mathcal{X} \rightarrow \mathbb{R}_{\geq0}$. The function $\rho$ is the input to $\sigma$ and defined as
\begin{equation} \label{eqn:rho}
    \rho(y_1,y_2) := \frac{y_1-y_2}{w},
\end{equation}
where $y_1$ and $y_2$ are the $y$-coordinates of two distinct vehicles with $y_1 > y_2$. The function $\rho$ can be viewed as a percentage of lane width $w$.
Function $\sigma$ has the following properties:
\begin{subequations}\label{eqn:sigmaProp}
    \begin{align}
    \rho(y_1,y_2) \geq 0.9 & \Rightarrow \sigma(\rho(y_1,y_2)) \leq 0 \label{eqn:sigma1}\\
    \rho(y_1,y_2) \geq 0.5 & \Rightarrow \sigma(\rho(y_1,y_2)) \geq 0.9 \label{eqn:sigma2}\\
    \rho(y_1,y_2) \leq 0.3 & \Rightarrow 1 \leq \sigma(\rho(y_1,y_2)) \leq 1.01. \label{eqn:sigma3}
\end{align}
\end{subequations}
$\sigma$ can be viewed as a percentage of safety distance $\tau_D v_E$.
A feasible choice for $\sigma$ is a sigmoid function of the form
\begin{equation}
    \sigma(\rho(y_1,y_2)) = \frac{s_1}{1+e^{s_2(\rho(y_1,y_2)-s_3)}}-s_4,
    \label{eqn:sigma}
\end{equation}
with design parameters $s_i \in \mathbb{R}_{>0}$, $i = \{1,...,4\}$.
Then, the candidate CBFs for keeping a safe distance to vehicles $\pm$1F are
\begin{subequations}\label{eqn:b4b5}
\begin{align}
    b_6(\bm{x}) & = x_{-1F}-x_E-\tau_D v_E \cdot \sigma(\rho(y_E,y_{-1F})) \label{eqn:b6}\\
    b_7(\bm{x}) & = x_{+1F}-x_E-\tau_D v_E \cdot \sigma(\rho(y_{+1F},y_E)), \label{eqn:b7}
\end{align} 
\end{subequations}
\begin{lemma}
    The functions $b_{2}$, $b_{3}$, $b_{4}$ and $b_{5}$ are HOCBFs with relative degree $r=2$. $b_6$ and $b_7$ are CBFs. If $-\frac{\pi}{2} < \Psi_E < \frac{\pi}{2}$, then $b_1$ is a CBF.
    \label{thm:HOCBF}
\end{lemma}
\begin{proof}
For $b_{2}$, $L_g b(t,\bm{x}) = 0$ when $\Psi_E = 0$, hence $r\neq 1$ based on Def. \ref{def:reldeg}. Then, $L_g L_f b(t,\bm{x}) \neq 0$ $\forall \bm{x}$, hence $r=2$ and based on Def. \ref{def:HOCBF}, $b_{2}$ is a HOCBF. Analogously for $b_{3}$, $b_{4}$ and $b_{5}$. For $b_6$ and $b_7$, $L_g b(t,\bm{x}) \neq 0$, $\forall \bm{x} \in \mathcal{X}$. For $b_1$ only if $-\frac{\pi}{2} < \Psi_E < \frac{\pi}{2}$. Hence, they are CBFs based on Def. \ref{def:HOCBF}.
\end{proof}
The assumption $-\frac{\pi}{2} < \Psi_E < \frac{\pi}{2}$ is reasonable for vehicles on a highway, since vehicles are not allowed to turn there. 
The (HO)CBFs are summarized in Table \ref{tab:barrierfct}.

\subsection{Controller}
As in \eqref{eqn:mult}, we summarize HOCBFs and HOCLFs as stack vectors
\begin{subequations}
\begin{align}
    \psi(\bm{x}) & := \begin{bmatrix}
    \dot{b}_1(\bm{x})+\gamma(b_1(\bm{x}))\\
    \psi_{2,r}(\bm{x})\\
    \psi_{3,r}(\bm{x})\\
    \psi_{4,r}(\bm{x})\\
    \psi_{5,r}(\bm{x})\\
    \dot{b}_6(\bm{x})+\gamma(b_6(\bm{x}))\\
    \dot{b}_7(\bm{x})+\gamma(b_7(\bm{x}))\\
    \end{bmatrix},
    \label{eqn:psiQP}\\
    \eta(\bm{x},\delta_{\omega}) & := 
    \eta_{1}(\bm{x}) + \delta_{\omega} \label{eqn:etaQP}
\end{align}
\end{subequations}
where functions $\psi_{k,r}$ for HOCBFs $b_k$, $k\in\{2,3,4,5\}$, are defined by \eqref{eqn:psiHOBF}, $r$ denotes the relative degree and it is $r=2$ (cf. Lemma~\ref{thm:HOCBF}). Analogously, function $\eta_1$ for HOCLF $V$ is defined by \eqref{eqn:etaHOLF}.
The control inputs $[v_E,\omega_E]^T$ are determined by the QP
\begin{subequations}
\label{eqn:QPboth}
\begin{align}
\min_{v_E,\:\delta_{v},\:\omega_E,\:\delta_{\omega}} \quad & 
H_v v_E^2 + H_{\omega} \omega_E+p_v \cdot \delta_{v}^2+p_{\omega} \cdot \delta_{\omega}^2\\
\textrm{s.t.} \quad 
  & \psi(\bm{x}) \geq 0 \label{eqn:QP2c}\\
  & \eta(\bm{x},\delta_{\omega}) \geq 0 \label{eqn:QP2d}\\
  & (v_E-v_{\text{ref}})+\delta_v  =0 \label{eqn:QP2b} \\
  & v_{E,\text{min}}\leq v_E \leq v_{E,\text{max}} \label{eqn:QP2e} \\
  & \omega_{E,\text{min}}\leq\omega_E \leq \omega_{E,\text{max}} \label{eqn:QP2f}
\end{align}
\end{subequations}
where $H_v$, $H_{\omega}$, $p_v$, $p_{\omega} \in \mathbb{R}_{\geq 0}$.
By constraint \eqref{eqn:QP2b}, the reference velocity $v_{\text{ref}}$ is tracked as closely as safety constraints \eqref{eqn:QP2c} admit. Constraints \eqref{eqn:QP2e}-\eqref{eqn:QP2f} are input constraints.

\subsection{Theoretical guarantees}
Consider QP~\eqref{eqn:QPboth}. As an immediate consequence of Corollary \ref{cor:QP}, any locally Lipschitz continuous controller $\bm{u}(\bm{x}) \in \mathcal{U}_{\text{HOCLF}}(\bm{x})$ to vehicle dynamics~\eqref{eqn:VehDyn} guarantees asymptotic stability at $y_E = y_{\text{ref}}$ if $y_{\text{ref}} \in \bigcap_{k=1}^{5}\mathcal{C}_i$ and $[\delta_v,\delta_{\omega}]=0$.
Furthermore, we can show the following.
\vspace{0.2cm}
\begin{theorem}\label{thm:FwdInvGua}
Let $b_i$, $i = 1,...,7$, be (HO)CBFs defined as in Table \ref{tab:barrierfct}. Then, $\mathcal{C} = \bigcap_{k=1}^7 \mathcal{C}_k$ is non-empty. Moreover, any locally Lipschitz continuous control $\bm{u}(\bm{x}) \in \mathcal{U}_{\text{HOCBF}}(\bm{x})$ to the system \eqref{eqn:VehDyn} renders $\mathcal{C}$ forward invariant and asymptotically stable.
\end{theorem}
\begin{proof}
At first, we show that $\mathcal{C} \neq \emptyset$. We start by considering all constraints on the $x$-coordinate. From $b_1\geq 0 $, we obtain that $x_{0F}-\tau_D v_E\geq x_E$ and it follows $\lbrace \bm{x} \:|\: x_E \in (-\infty,x_{0F}-\tau_D v_E]\rbrace \subseteq \mathcal{C}_1$. Similarly, we obtain  from $b_6\geq0$ that $x_{-1F}-\tau_D v_E \cdot \sigma \geq x_E$ and it follows $\lbrace \bm{x} \: | \: x_E \in (-\infty,x_{-1F}-\tau_D v_E \cdot \sigma]\rbrace \subseteq \mathcal{C}_6$. Analogously, we obtain from $b_7\geq 0$ that $\lbrace \bm{x} \: | \: x_E \in (-\infty,x_{+1F}-\tau_D v_E \cdot \sigma]\rbrace \subseteq \mathcal{C}_7$. Consequently, as $x_E$ is not lower bounded, $\bigcap_{k=1,6,7}\mathcal{C}_k \neq \emptyset$.

Next, we consider the constraints on the $y$-coordinate. From $b_2$, we obtain that
\begin{align} \label{eqn:Prb2}
    & b_2(\bm{x})  = y_E- y_{\text{min}}^\ell +\underbrace{w\cdot \lambda}_{\geq 0} \geq y_E - y_{\text{min}}^\ell 
\end{align}
and thus 
\begin{align} 
    \label{eqn:PrC2}
    \begin{split}
        \lbrace \bm{x} \: | \: y_E \in [y_{\text{min}}^\ell, y_{\text{max}}^\ell]\rbrace &= \{\bm{x}\:|\:y_E-y_{\text{min}}^\ell \geq 0\} \\
    &\subseteq \{\bm{x}\:|\:b_2(\bm{x})\geq0\} = \mathcal{C}_2.
    \end{split}
\end{align} 
By proceeding analogously for $b_3$, $b_4$, $b_5$, we obtain $\lbrace \bm{x} \: | \: y_E \in [y_{\text{min}}^\ell, y_{\text{max}}^\ell]\rbrace \subseteq \bigcap_{k=2}^5\mathcal{C}_{k} $. Altogether, we have $\mathcal{C} = \bigcap_{k=1}^7 \mathcal{C}_k \neq \emptyset$.

At last, as functions $b_k$, $k = 1,...,7$, are (HO)CBFs due to Lemma \ref{thm:HOCBF}, $\mathcal{C}$ is forward invariant and asymptotically stable according to Corollary \ref{cor:QP}. 
\end{proof}

\subsection{Coordination principle}

\begin{figure}
    \centering
    \hspace{-0.35cm}
    \subfloat[]{
         \includegraphics[width=4.35cm]{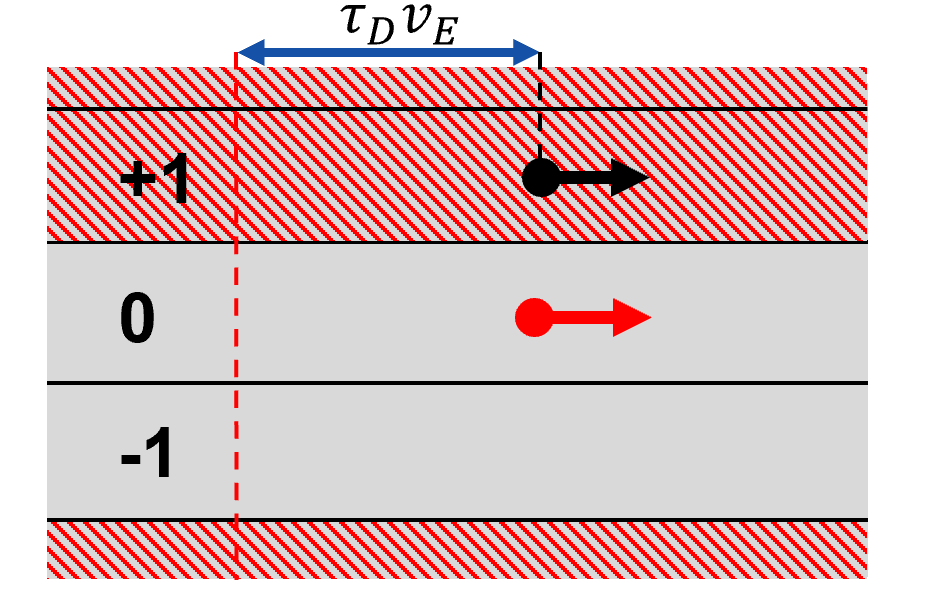}
         \label{fig:B1}
    }
	\hfill
    \subfloat[]{
         \includegraphics[width=3.75cm]{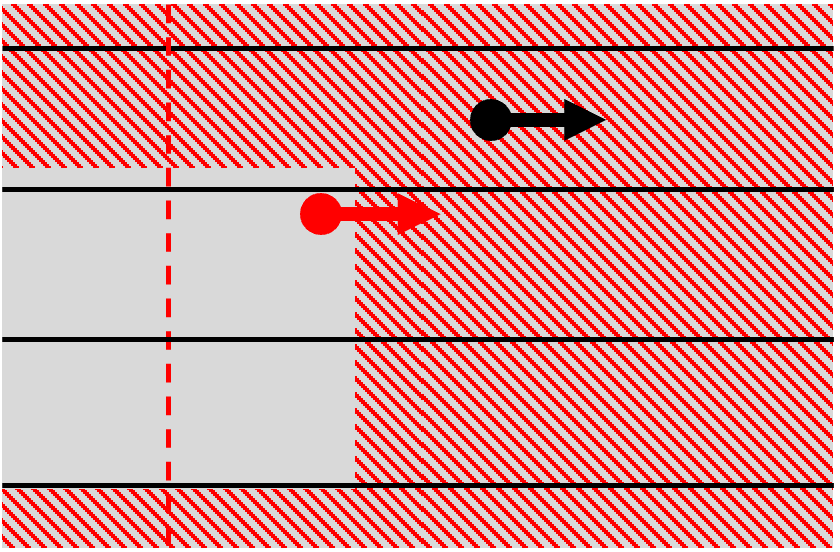}
         \label{fig:B2}        
    }
	\\
    \subfloat[]{
         \includegraphics[width=3.75cm]{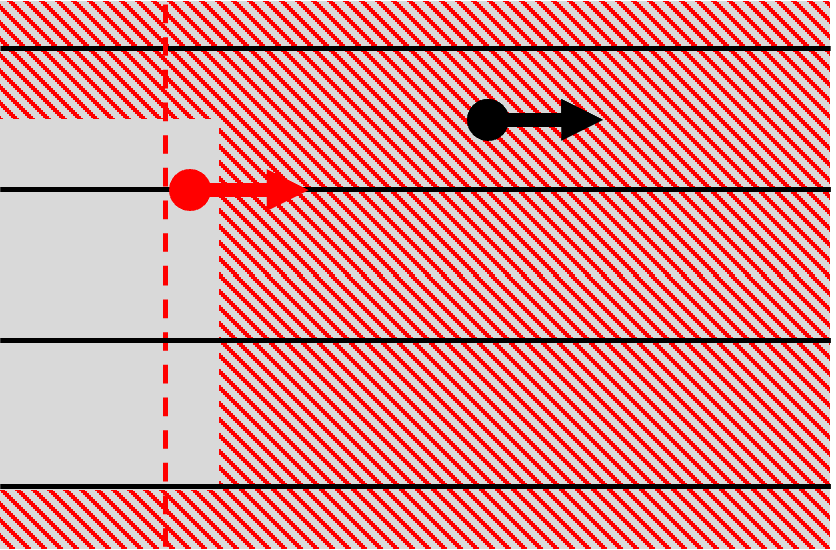}
         \label{fig:B3}
    }
	\hfill
    \subfloat[]{
         \includegraphics[width=3.75cm]{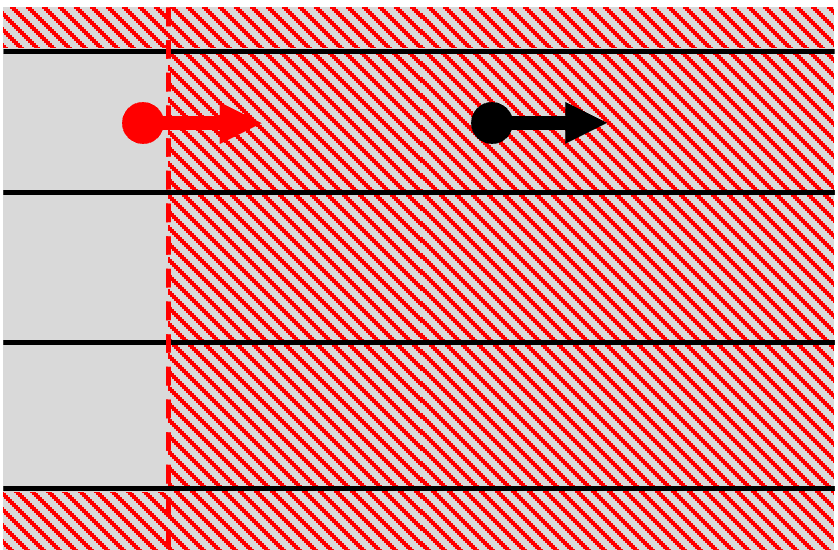}
         \label{fig:B4}        
    }
    \caption{Safe region of the ego-vehicle (red) during lane switching.}
    \label{fig:BDescrip}
\end{figure}
In order to illustrate the concept of vehicle coordination, a lane change manoeuvre is shown in Fig.~\ref{fig:BDescrip}. The ego-vehicle (red) changes from lane $\ell = 2$ (index 0) to lane~3 (index +1). Another vehicle (black) denoted by +1F is ahead of the ego vehicle on the neighboring lane with $y_{+1F} = (\ell+1)w$. The red area denotes the unsafe region for the ego-vehicle, i.e., where $b_k<0$ for some $k=\{1,...,7\}$.

We define the safe set to $b_k$ as $\mathcal{C}_k := \{\bm{x}=[x_E, y_E,\Psi_E,\\ \bm{x}_{i\in \mathcal{N}}]^T\:|\:b_k(\bm{x})\geq 0\}$ and $\mathcal{C} := \bigcap_{k=1}^7 \mathcal{C}_k$. 

At first observe that the ego-vehicle can move freely on its lane independently of the position of neighbouring vehicles (Fig.~\ref{fig:B1}), i.e., $\lbrace \bm{x} \: | \: y_E \in [y_{\text{min}}^\ell, y_{\text{max}}^\ell]\rbrace \subseteq \mathcal{C}$. This has already been shown in the proof of Thm.~\ref{thm:FwdInvGua} in  \eqref{eqn:Prb2}-\eqref{eqn:PrC2}. 
Thus, we conclude that the ego-vehicle can always manoeuvre to $y_E = y_{\text{max}}^\ell$ independently of the states of the neighboring vehicles. 

If the ego-vehicle manoeuvres to $y_E = y_{\text{max}}^\ell$, then $y_{+1F}-y_E = y_{+1F}-y_{\text{max}}^\ell = (\ell+1)w - (\ell+\frac{1}{2})w= \frac{w}{2}$, and it follows from~\eqref{eqn:sigma2} that $\sigma(\rho(y_{+1F},y_E)) \geq 0.9$. Furthermore, we obtain from~$b_7 \geq 0$ also that $x_{+1F}-x_E \geq 0.9 \,\tau_D v_E$, which leads to $\lambda(\theta(x_{+1F},x_E,v_E)) = 0.5$ according to~\eqref{eqn:lambda2}. Consequently, we have 
\begin{align*}
    & b_5(\bm{x})  = w \cdot \underbrace{\lambda}_{= 0.5}+y_{\text{max}}^\ell-y_E
\end{align*}
and thus
\begin{align*}
    \{\bm{x}|&y_E \in (-\infty,\frac{w}{2} +y_{\text{max}}^\ell], \;  x_{+1F}-x_E \geq 0.9 \,\tau_D v_E\} \\ & = \{\bm{x}\:|\: 0.5\,w +y_{\text{max}}^\ell\geq y_E,\:  x_{+1F}-0.9\,\tau_D v_E \geq x_E\} \\ &\subseteq \{\bm{x}\:|\:b_5(\bm{x})\geq0, \: b_7(\bm{x})\geq 0\} = \mathcal{C}_5 \cap \mathcal{C}_7.
\end{align*}
Hence, we can conclude that once the ego-vehicle reaches $y_E = y_{\text{max}}^\ell$, the center of lane~3, which is $(\ell+1)w = 3w = \frac{w}{2} +y_{\text{max}}^\ell $, is also contained in the safe set $\mathcal{C}$. In Fig.~\ref{fig:B2}-\ref{fig:B4}, the vehicle eventually switches to lane 3.

\enlargethispage{\baselineskip}

In summary, the functions $\lambda$ and $\sigma$ couple and coordinate the safe regions in $x$- and $y$-direction and are therefore called \emph{coordination functions}. Since vehicles are not coordinated by a high level traffic coordinator this control approach is decentralized.

\section{SIMULATION}
\label{sec:Sim}
The proposed controller is validated in a numerical simulation. We consider a highway with two lanes of width~$w$ with 2-3 identical autonomous vehicles. The vehicles behave according to their kinematic model as given in \eqref{eqn:VehDyn}. All vehicles calculate their input via the same CLF-CBF-QP~\eqref{eqn:QPboth} with parameters as given in Table~\ref{tab:ConstVal}.

A video illustrates the simulation results\footnote{https://www.youtube.com/watch?v=OxPSGhFoq2o}. The video shows three scenarios. In each of the scenarios, a further task is added: Whereas in the first scenario, a vehicle only needs to follow another vehicle in a safe distance, in scenario two a lane change is added. In the third scenario, the neighboring vehicles need to additionally open a gap before the lane change is completed. Due to space limitations, we only show the simulation results of
the second scenario in Fig.~\ref{fig:Scen2}, where a vehicle switches the lane in front of a neighbouring vehicle.
The simulation is implemented in Matlab and runs on an Intel Core i7 1.3 GHz with 32 GB RAM. The optimization is solved using the function \textit{fmincon}. The average computation time for the control input is 0.096 sec.

\begin{table}[t]
    \centering
\caption{Simulation parameters.
}
\label{tab:ConstVal}
    \begin{tabular}{|c|c|c|c|c|c|}
    	\hline
    	$a_1$ &  $a_2$ & $a_3$ & $\beta_1$ & $\beta_2$ & $\beta_3$\\ \hline
    	234.14 & -0.872 & 0.4949 & 1209.2 & -0.9962 & 0.01\\ \hline
    	\multicolumn{6}{c}{}\\ \hline
    	$s_1$ & $s_2$ & $s_3$ & $s_4$ & $\tau_D$ & $r_S$\\ \hline
    	1.03 & 16 & 0.64 & 0.02 & 0.9 $[s]$ & 100 $[m]$\\\hline
    	\multicolumn{6}{c}{}\\ \cline{2-5}
    	\multicolumn{1}{c|}{} & $H_v$ & $H_\omega$ & $p_v$ & $p_\omega$ & \multicolumn{1}{|c}{}\\ \cline{2-5}
    	\multicolumn{1}{c|}{} & 1 & 70,000 & 1e9 & 1e9 & \multicolumn{1}{|c}{} \\\cline{2-5}
    \end{tabular}
\end{table}

\begin{figure}[tb]
    \centering
    \includegraphics[width=0.9\columnwidth]{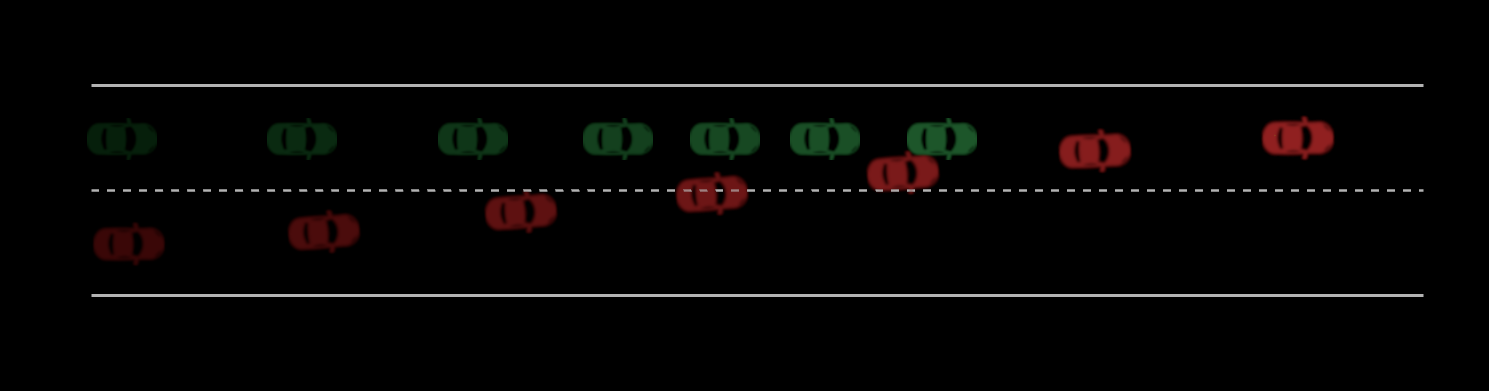}
    \caption{Scenario 2: adaptive cruise control with lane switching. The vehicle shadows represent time instances $t = 0,...,6$.}
    \label{fig:Scen2}
\end{figure}

\section{CONCLUSION}
\label{sec:Con}
In this work, a decentralized controller based on a CLF-CBF-QP is presented which does not require to switch between several controllers.
The controller enables autonomous driving on a lane with adaptive cruise control, and allows for lane switching without collisions. 
The novelty of the approach is that the vehicles indicate their objective by manoeuvres. To this end, we introduced coordination functions for coordinating the vehicles' safe regions.

\balance







\bibliographystyle{ifacconf}

\end{document}